\documentclass[a4paper, amsfonts, amssymb, amsmath, reprint, showkeys, nofootinbib, twoside,notitlepage,onecolumn]{revtex4-1}

\def\prd{{Phys.~Rev.~D}}        
\def\sovast{{Soviet~Ast.}}      

\usepackage{amsmath,amstext}
\usepackage[T1]{fontenc}
\usepackage{amssymb}
\usepackage{graphicx}
\usepackage{ae,aecompl}

\DeclareFontFamily{OT1}{pzc}{}
\DeclareFontShape{OT1}{pzc}{m}{it}{<-> s * [1.10] pzcmi7t}{}
\DeclareMathAlphabet{\mathpzc}{OT1}{pzc}{m}{it}

\usepackage{hyperref}
\usepackage{amsmath}
\usepackage{amssymb}
\usepackage{mathtools}
\usepackage{bm}
\usepackage{cleveref}
\usepackage{tensor}
\usepackage{braket}
\usepackage{enumitem}
\usepackage{mhchem}
\usepackage{amsthm}
\usepackage{nccmath}
\usepackage{mathrsfs}
\usepackage{color}

\newcommand{\spc}{\quad \quad \quad}

\def\be{\begin{equation}}
\def\ee{\end{equation}}
\def\beq{\begin{eqnarray}}
\def\eeq{\end{eqnarray}}

\theoremstyle{definition}

\theoremstyle{theorem}
\newtheorem{theorem}{Theorem}

\begin{document}
\title{Plasma oscillations within Israel-Stewart theory}
\author{L.~Gavassino}
\affiliation{Department of Applied Mathematics and Theoretical Physics, University of Cambridge, Wilberforce Road, Cambridge CB3 0WA, United Kingdom}

\begin{abstract}
It is well known that, at zero wavenumber, the non-hydrodynamic frequencies of uncharged kinetic theory are purely imaginary. On the other hand, it was recently shown that, in resistive magnetohydrodynamics, the interplay between the Israel-Stewart relaxation equation and the Amp\`{e}re-Maxwell law can give rise to a pair of oscillating non-hydrodynamic modes. In this work, we analyze this phenomenon in detail. We first demonstrate that these oscillatory modes are exact solutions of the Drude model, corresponding to ordinary plasma oscillations. We then invoke the Onsager–Casimir principle to explain that their oscillatory nature reflects the distinct PT-transformation properties of the degrees of freedom: the distribution function is even, while the electric field is odd. Finally, we establish that, in a kinetic theory of charged particles, there can be at most one such pair of oscillatory modes per spatial dimension, while all other modes still must sit on the imaginary axis.
\end{abstract} 
\maketitle

\section{Introduction}
\vspace{-0.3cm}

It has been firmly established that ensuring causality in a hydrodynamic theory necessitates the inclusion of non-hydrodynamic modes \cite{HellerBounds2022ejw,GavassinoDisperisons2023mad}, i.e. of additional excitations whose dispersion relation $\omega(k)$ is gapped at low momentum\footnote{The term \textit{non-hydrodynamic} comes from the fact that ordinary Navier-Stokes hydrodynamics has no gapped modes. In fact, the Navier-Stokes equations of motion are all conservation laws of the form $\partial_t \Psi{=}-\partial_j F^j(\Psi,\partial_l \Psi)$, which gives $\omega=k f(k)$, and thus  $\omega(k{=}0)\equiv 0$.}. From a mathematical standpoint, this requirement is completely analogous to the necessity of allowing for antiparticle modes within causal classical field theories \cite{GavassinoDisperisons2023mad}. As a result, substantial recent effort has been directed toward characterizing the distinctive properties of the non-hydrodynamic spectrum in realistic microscopic theories \cite{KovtunHolography2005,Denicol_Relaxation_2011,Moore:2018mma,Andrade:2019zey,BAGGIOLI20201,Wagner:2023jgq,GavassinoGapless:2024rck,Brants:2024wrx,RochaBranchcut:2024cge,AhnBaggioli:2025odk,Hernandez:2025zxw,Bajec:2025dqm}, and toward establishing systematic strategies for incorporating these properties into (quasi)hydrodynamic models \cite{Israel_Stewart_1979,Denicol2012Boltzmann,Heller2014,Grozdanov2019,BulkGavassino,GavassinoNonHydro2022,GavassinoGENERIC:2022isg,GavassinoBurgers:2023eoz,GavassinoUniveraalityI2023odx,Denicol:2024cpj}. The general picture one gets from these studies is as follows (see figure \ref{fig:typical}): for systems described by kinetic theory, the non-hydrodynamic excitations at $k{=}0$ lie on the imaginary axis (they describe pure relaxation), and Israel-Stewart theory offers a reasonable approximation of their behavior. For systems described holographically, instead, the non-hydrodynamic modes appear as conjugate pairs of the form $\omega\,{=}\,{-}ia\,{\pm}\, b$ (they undergo damped harmonic oscillations), which forces us to move outside of the standard Israel-Stewart framework. 

Recently, an intriguing exception to this conventional wisdom has been discovered \cite{Dash:2022xkz}. Specifically, when Ohm’s law ($J=\sigma E$) is modified by adding an Israel–Stewart-type relaxation term characteristic of kinetic theories ($\tau \dot{J}+J=\sigma E$), the resulting magnetohydrodynamic system supports a pair of complex-conjugate non-hydrodynamic modes at $k\,{=}\,0$ (for certain choices of $\tau$ and $\sigma$). In other words, introducing a transient response of the current leads not to relaxation, but to oscillations. This raises some questions. Are these oscillations physical, and do they have a correspondent in metals? How many oscillating modes can exist in a kinetic theory of charged particles? What is the physical mechanism that causes charged excitations to depart from the imaginary axis? Here, we will answer these questions.

We adopt the metric signature $(-,+,+,+)$, and work in natural units $c=\hbar=k_B=1$ and $e^2=4\pi/137$.

\begin{figure}[h!]
    \centering
\includegraphics[width=0.44\linewidth]{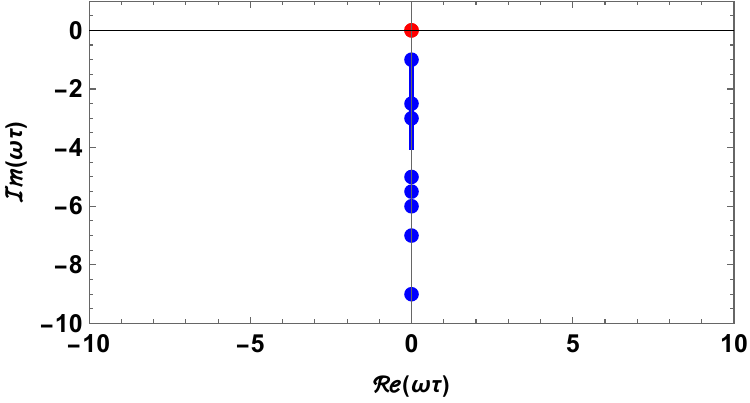}
\includegraphics[width=0.44\linewidth]{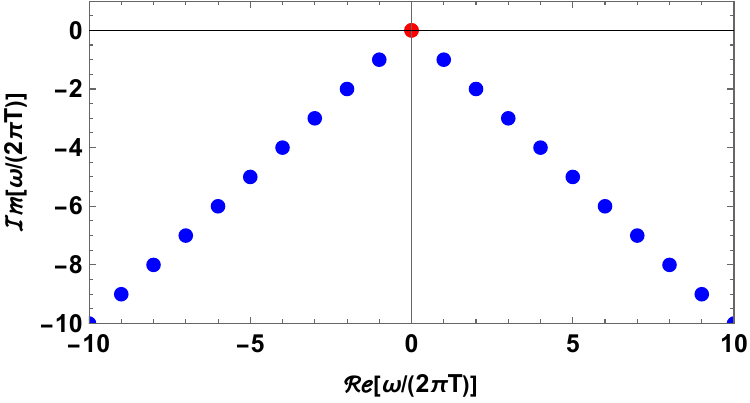}
\caption{Hydrodynamic (red) and non-hydrodynamic (blue) modes of a microscopic theory at $k=0$ according to the conventional wisdom. Left panel: Typical excitation spectrum of a kinetic theory. All the non-hydrodynamic frequencies are scattered along the imaginary axis (some may form a continuum \cite{Moore:2018mma}, thin blue line), and their magnitude is of the order of the particle mean free time $\tau$. Right panel: Typical excitation spectrum of a holographic theory. The non-hydrodynamic relaxation rates $-i\omega$ come in (often regularly spaced) complex conjugate pairs, and their magnitude is of the order of the temperature $T$.}
    \label{fig:typical}
\end{figure}

\newpage
\section{Derivation from kinetic theory}

Our first task is to determine whether the oscillating modes identified by \cite{Dash:2022xkz} represent a limitation of Israel-Stewart theory, or if they correspond to a real physical effect occurring in kinetic theory. We will address this question using a simple Drude-type kinetic model in which the positive charges are fixed, while the electrons constitute a non-relativistic (degenerate) gas.

\subsection{Step-by-step calculation}\label{stepbystep}

We consider an infinite conductor at rest, and we write Maxwell's equations\footnote{Strictly speaking, one should write down the Maxwell equations using the fields in matter, in order to account for bound charges and currents. Here, we ignore those effects for clarity, since they do not introduce any qualitative change.} (in Heaviside-Lorentz units \cite[\S 54]{Srednicki_2007}):
\begin{equation}\label{maxwellone}
\begin{split}
& \partial_j E^j=\rho \, , \\
& \partial_j B^j=0 \, , \\
& (\nabla \times E)^j =-\partial_t B^j \, , \\
& (\nabla \times B)^j =J^j +\partial_t E^j \, , \\
\end{split}
\end{equation}
with $E^j$ and $B^j$ the electric and magnetic field, while $\rho$ and $J^j$ are the charge density and current of the conductor.
Given that we are interested in the behavior of the conductor at vanishing wavenumber, we can impose $\partial_j = 0$, which immediately implies $\rho = 0$ (charge neutrality) and $\partial_t B^j = 0$. The latter equation tells us that the magnetic field is uniform both in space and in time, and so we can set $B^j = 0$, for convenience. Therefore, the only relevant Maxwell equation left is the Amp\`{e}re-Maxwell law:
\begin{equation}\label{amperemaxwell}
\partial_t E^j=-J^j \, .
\end{equation}
In Drude's model, the electric current is given by the kinetic formula
\begin{equation}
J^j =-e \int \dfrac{2d^3 p}{(2\pi)^3} \,  f \, \dfrac{p^j}{m} \, ,
\end{equation}
where $e>0$ and $m>0$ are the fundamental charge unit and the electron mass, while $f=f(t,p^k)\in [0,1]$ is the electron distribution function, i.e. the average occupation number of the single-electron state with momentum $p^k$. Under our aforementioned assumptions (i.e. $\partial_j=B^j=0$), the function $f(t,p^k)$ obeys the Boltzmann-Vlasov equation
\begin{equation}\label{boltzmannVlasov}
\dfrac{\partial f}{\partial t}-eE^k \dfrac{\partial f}{\partial p^k} =\dfrac{\langle f\rangle -f}{\tau} \, ,
\end{equation}
with $\tau$ the electron mean free time, which is assumed to be momentum-independent for simplicity. The isotropic function $\langle f\rangle=\langle f\rangle (p^k p_k)$ is the average of $f$ over all angles. This choice of collision integral corresponds to assuming that the electrons relax to equilibrium thanks to elastic isotropic scattering with the ion lattice.

Now, the key observation is that, if we multiply both sides of \eqref{boltzmannVlasov} by $p^j/m$, and integrate over all momenta, we recover the Israel-Stewart relaxation-type equation as an \textit{exact} kinetic identity: 
\begin{equation}\label{IsraelStewart}
\tau \partial_t J^j+J^j=\sigma E^j \, ,
\end{equation}
with $\sigma=e^2 n_e\tau/m$ the electric conductivity of Drude's model ($n_e=\text{``electron number density''}$). Therefore, by combining \eqref{amperemaxwell} with \eqref{IsraelStewart}, we arrive at the equation of motion for the current that was found by \cite{Dash:2022xkz}:
\begin{equation}\label{plasmafrequency}
\partial^2_t J^j +\dfrac{1}{\tau} \partial_t J^j +\omega_p^2 J^j =0 \, .
\end{equation}
This has the form of a damped harmonic oscillator, whose frequency $\omega_p\equiv \sqrt{\sigma/\tau}=\sqrt{e^2 n_e/m}$ is nothing but the plasma frequency of the medium. If we assume a time dependence $J^j\propto e^{-i\omega t}$, we obtain
\begin{equation}\label{omtau}
2\omega\tau=-i\pm i\sqrt{1-(2\omega_p \tau)^2} \, ,
\end{equation}
which tells that, if $\omega_p\tau>1/2$ (underdamped case), there are indeed two oscillating modes.

In summary: The oscillating non-hydrodynamic modes \eqref{omtau} found by \cite{Dash:2022xkz} are not mere artifacts of the Israel-Stewart approximation. Instead, they are exact solutions of the Boltzmann-Vlasov-Maxwell system \eqref{amperemaxwell}-\eqref{boltzmannVlasov}.

\subsection{Quantitative estimates}

For the non-hydrodynamic frequencies \eqref{omtau} to have a real part (i.e. for the modes to be oscillating), we need that $\omega_p \tau>1/2$. In a typical metal (e.g. gold, or copper \cite{Youn2007ExtendedDrude}) at room temperature, we have $\omega_p\tau \sim 100$ \cite[\S 1]{ashcroft1976solid}, so these modes do oscillate. A similar result holds for white dwarfs and neutron stars, where $\omega_p\tau \sim 10^3$ \cite{YakovlevUrpin1980}. In astrophysical plasmas, $\omega_p \tau$ can range from $10^7$ in the solar corona to $10^{15}$ in accretion disks \cite{NRLFormulary,Draine2011,Sarazin1988,Aschwanden2005,KivelsonRussell1995,Frank2002,NarayanYi1995,Quataert1998}, so also these systems exhibit oscillatory dynamics. In the quark-gluon plasma formed in heavy-ion collisions, the exact value of $\omega_p \tau$ is not known, but the authors of \cite{Dash:2022xkz} assume it to be of order 1. Hence, there may or may not be oscillations, depending on the details of the interactions.

\subsection{Interpretation of the oscillating modes as plasmons}

Given the appearance of the plasma frequency in equation \eqref{plasmafrequency}, one might be tempted to interpret the modes \eqref{omtau} as ordinary plasma oscillations \cite[\S 1.5]{Fitzpatrick2020}, or ``plasmons'' \cite{PinesNozieres1958}. At the same time, this interpretation may seem questionable, since plasmons are typically understood as oscillations of the electron density about charge neutrality, whereas our modes maintain $\rho=0$ everywhere. In what follows, we demonstrate that the modes \eqref{omtau} are indeed plasmons, with equation \eqref{plasmafrequency} capturing the bulk behavior of their evolution, as observed away from boundaries.

Let us revisit the conventional plasmon picture \cite[\S 1]{ashcroft1976solid} (see figure \ref{fig:Plasmonmodel}). We consider a finite metallic sample of volume $V$, and we imagine rigidly displacing the electron gas along the $x$-direction by an amount $\Delta x$. This displacement induces a thin surface charge layer, which in turn generates an electric field equal to $E^x = e n_e \Delta x$ (just solve $\partial_x E^x=\rho$). Differentiating this equation with respect to time, assuming that the electron gas moves as a rigid block with time-dependent displacement $\Delta x(t)$, we obtain
\begin{equation}\label{capacitor}
\partial_t E^x = e n_e \Delta \dot{x}=-J^x \, .
\end{equation} 
As can be seen, we have already recovered equation \eqref{amperemaxwell}. Moving on, let us now determine the equation of motion of the baricenter of the electron gas. According to Newton's second law, we have
\begin{equation}\label{Newtonsecoindlaw}
M \Delta \Ddot{x}= F_{\text{friction}}+F_{\text{Coulomb}} .
\end{equation}
For the total mass of the gas, we can take $M=mn_e V$, and for the frictional force, we can take $F_{\text{friction}}=-M\Delta \Dot{x}/\tau$. In the limit of small $\Delta x$, the Coulomb force reads $F_{\text{Coulomb}}=-e E^x n_e V$, and equation \eqref{Newtonsecoindlaw} becomes
\begin{equation}\label{is}
\tau \Dot{J}^x+J^x=\sigma E^x \, .
\end{equation}
We have also recovered \eqref{IsraelStewart}. Combining \eqref{capacitor} and \eqref{is}, we are back to \eqref{plasmafrequency}.

Our calculation reveals a philosophically elegant point: The bulk behavior \eqref{plasmafrequency} is the same whether we view the electric field as produced by a surface charge (via Gauss' law) or, in an infinite medium, by a time-varying displacement field (via Amp\`{e}re-Maxwell). This equivalence follows from the fact that the time derivative of Gauss's law is identical to the divergence of the Amp\`{e}re-Maxwell law when we account for charge conservation. We can also view this equivalence from a causality perspective. An observer deep inside a large conductor cannot know whether the system is finite (with surface charge) or infinite (without surface charge), so the fields must behave the same either way.

\begin{figure}[b!]
    \centering
\includegraphics[width=0.4\linewidth]{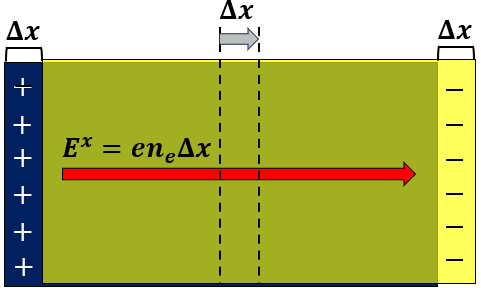}
\caption{A minimal model of plasma oscillations can be constructed by rigidly displacing the electron gas (yellow) in a finite conductor (blue), which produces two oppositely charged layers at the boundaries. The resulting electric field (red arrow) exerts a restoring force on the electrons. Including both the electron inertia and the electron-proton friction leads directly to the equation of a damped harmonic oscillator.}
    \label{fig:Plasmonmodel}
\end{figure}

\section{Onsager analysis}
\vspace{-0.2cm}

Having established that the oscillating modes are physical, we can move on to discuss some more advanced topics. This section is organized as follows. In subsection \ref{iiia}, we will explain why, in a kinetic theory of uncharged particles, all the modes sit on the imaginary axis at $k=0$. In subsection \ref{iiib}, we will explain why the introduction of a coupling with the electric field can give rise to oscillations also at $k=0$. In subsection \ref{iiic}, we will determine how many oscillating non-hydrodynamic modes a generic kinetic theory of charged particles can sustain at $k=0$. We will answer all these questions within the framework of the Onsager-Casimir theory of irreversible processes \cite{Onsager_1931,Onsager_Casimir}.

\vspace{-0.2cm}
\subsection{Why uncharged kinetic theory is purely relaxational}\label{iiia}
\vspace{-0.2cm}

Let us briefly recap the basics of Onsager's theory. The most probable macrostate of a thermodynamic system is the state that minimizes some dimensionless thermodynamic quantity $\Omega$ ($=-\text{``entropy''}$, for isolated systems). Let $\Psi$ be a list of linearised real degrees of freedom that we use to parameterize small fluctuations around this most probable state. Then, the probability for the system to be found in the macrostate $\Psi$ is given by $\mathcal{P}(\Psi)\propto e^{-\Delta\Omega(\Psi)}$. In the Gaussian approximation, $\Delta \Omega$ can be expanded to second order in fluctuations, giving $\Delta \Omega(\Psi)=\Psi^T \mathbb{M}\Psi/2+\mathcal{O}(\Psi^3)$, for some symmetric matrix $\mathbb{M}$. With an appropriate change of basis in the $\Psi-$space (we treat $\Psi$ as a column vector), we can always set $\mathbb{M}=1$, so that
\begin{equation}\label{identif}
\Delta\Omega(\Psi)=\dfrac{\Psi^T\Psi}{2}+\mathcal{O}(\Psi^3) \, .
\end{equation}
Now, it was shown \cite{Onsager_1931,Onsager_Casimir} that, in this particular $\Psi-$basis, if the linearised dynamics can be written in the form
\begin{equation}\label{dotpsiLpsi}
\dot{\Psi} = -\mathcal{L}\Psi ,
\end{equation}
and if both the microscopic theory and the equilibrium state are invariant under a discrete symmetry transformation $\Psi \rightarrow \epsilon\Psi$ that involves time reversal (such as T-, PT-, or CPT-symmetry), with $\epsilon^2=1$, then the real matrix $\mathcal{L}$ fulfills the following symmetry relation \cite{Krommes1993}:
\begin{equation}\label{Onsager}
\mathcal{L} = \epsilon \mathcal{L}^T\epsilon^T \,.
\end{equation}
In Appendix \ref{apponsager}, we provide a quick derivation of this well-known result.

Let us now examine the implications within kinetic theory. We take as our fundamental symmetry $\epsilon$ the combined action of parity and time reversal (``PT-symmetry''), which preserves most interactions \cite[\S 2.6]{weinbergQFT_1995}. Remarkably, this symmetry leaves the distribution function $f(p^k)$ unchanged, since under PT we have $p^k \rightarrow p^k$. Therefore, if we take $\Psi$ to be the list of moments of the distribution function, we immediately obtain $\epsilon = 1$, so that \eqref{Onsager} simplifies to $\mathcal{L} = \mathcal{L}^T$ \cite{GavassinoGapless:2024rck,GavassinoConvergence:2024xwf}. It follows that the spectrum of $\mathcal{L}$ is real. Consequently, solving \eqref{dotpsiLpsi} with the ansatz $\Psi \propto e^{-i\omega t}$,
\begin{equation}
-i\omega \Psi =- \mathcal{L} \Psi \, ,
\end{equation}
can only yield solutions with $\omega \in i\mathbb{R}$, in agreement with the general behavior shown in figure \ref{fig:typical} (left panel). Notably, this same argument applies to all transient fluid theories (such as Israel-Stewart \cite{Israel_Stewart_1979} or divergence-type theories \cite{Liu1986}), whose degrees of freedom (energy density, flow velocity, stresses, etc.) are also moments of the distribution function, all even under PT symmetry, resulting again in purely imaginary non-hydrodynamic frequencies \cite{GavassinoSymmetric2022nff}.

We stress that the identity $\epsilon = 1$ is strictly valid only at $k = 0$. When the distribution function depends on the position, PT symmetry acts as $f(x^k, p^k) \to f(-x^k, p^k)$. To implement the Onsager-Casimir principle in this setting, one can expand $f$ in Fourier sine and cosine modes, for example, $f(x,p^k) = f_s(p^k)\sin(kx) + f_c(p^k)\cos(kx)$. In this decomposition, $f_s$ is PT-odd and $f_c$ is PT-even. Equation \eqref{Onsager} then predicts an antisymmetric coupling between $f_s$ and $f_c$, which enables the frequency of, e.g., sound waves to acquire a real part \cite{GavassinoSymmetric2022nff}.

\vspace{-0.2cm}
\subsection{Why the electric field induces oscillations}\label{iiib}
\vspace{-0.2cm}

Let us now extend our kinetic model by coupling it with the electromagnetic field. The resulting Boltzmann-Vlasov-Maxwell model has degrees of freedom $\Psi_{\text{BVM}}=(\Psi,E^j,B^j)^T$, where $\Psi$ is still the list of moments of the distribution function, which are PT-even, while $E^j$ and $B^j$ are the electric and magnetic fields, which are PT-odd. Thus, we have that $\epsilon=\text{diag}(+1,-1,-1)$, with ``$1$'' being the appropriate identity matrices. Now, electrodynamics is PT-invariant, and so is the state of thermodynamic equilibrium, provided that it carries no average magnetic field. Moreover, since $\Omega$ is PT-even, there can be no cross terms $\Psi E$ and $\Psi B$ in $\Delta \Omega$. The cross term $EB$ is also forbidden because $\Omega$ and $E$ are T-even, while $B$ is T-odd. Hence, we have $\Delta\Omega=\Psi^T \Psi/2+E^j E_j/2+B^j B_j/2=\Psi^T_{\text{BVM}}\Psi_{\text{BVM}}/2$ (in some units), and we can apply the Onsager-Casimir principle to the evolution of $\Psi_{\text{BVM}}$ itself. The result is
\newpage
\begin{equation}
\begin{split}
\Dot{\Psi}={}& -\mathcal{L}\Psi +{\sum}_j\mathcal{J}_j E^j +{\sum}_j \mathcal{K}_j B^j \spc (\text{with }\mathcal{L}=\mathcal{L}^T)\, , \\
\Dot{E}_j ={}&-\mathcal{J}_j^T\Psi -\gamma_{EE} E_j -\gamma_{EB}B_j \, , \\
\Dot{B}_j ={}&-\mathcal{K}_j^T\Psi -\gamma_{EB} E_j -\gamma_{BB}B_j \, . \\
\end{split}
\end{equation}
This is the most general system of equations compatible with the Onsager-Casimir principle. If we further require this system to be consistent with Maxwell's equations \eqref{maxwellone} (in particular, the last two, for $\partial_j=0$), we are left with
\begin{equation}\label{langle}
\begin{split}
\Dot{\Psi}={}& -\mathcal{L}\Psi +{\sum}_j\mathcal{J}_j E^j\spc (\text{with }\mathcal{L}=\mathcal{L}^T)\, , \\
\Dot{E}_j ={}&-\mathcal{J}_j^T\Psi\, \, (\equiv -J^j) \, , \\
\Dot{B}_j ={}& 0 \, . \\
\end{split}
\end{equation}
As can be seen, the magnetic field completely decouples from the other degrees of freedom, and its presence does not affect the particle dynamics (in the linear regime!). By contrast, the electric field does couple with $\Psi$, and this coupling is antisymmetric. It is this antisymmetry that enables oscillatory dynamics. To see this, let us assume $\mathcal{L}\approx 0$, as is indeed the case for metals and astrophysical plasmas when studied on timescales of the order of $\omega_p^{-1}$. Moreover, let us assume that $E^j=(E^1,0,0)$. Then, contracting the first line of \eqref{langle} with $\mathcal{J}^T_1$ on the left, we obtain
\begin{equation}
\dfrac{d}{dt}
\begin{bmatrix}
J^1 \\
E^1 \\
\end{bmatrix}
=
\begin{bmatrix}
0 & \omega_p^2 \\
-1 & 0 \\
\end{bmatrix}
\begin{bmatrix}
J^1 \\
E^1 \\
\end{bmatrix}\, ,
\end{equation}
where we have defined $\omega_p^2\equiv \mathcal{J}_1^T \mathcal{J}_1\geq 0$. The solution is indeed oscillatory:
\begin{equation}
\begin{bmatrix}
J^1(t) \\
E^1(t) \\
\end{bmatrix}
=
\begin{bmatrix}
\cos(\omega_p t) & \omega_p \sin(\omega_p t) \\
-\sin(\omega_p t)/\omega_p & \cos(\omega_p t) \\
\end{bmatrix}
\begin{bmatrix}
J^1(0) \\
E^1(0) \\
\end{bmatrix}\, .
\end{equation}
For completeness, in Appendix \ref{drudeonsager}, we explicitly verify that the Drude model discussed in Subsection \ref{stepbystep} is consistent with the Onsager analysis presented above, i.e. it can be recast in the form \eqref{langle}.

Sidenote: Since transient fluid theories are always compatible with the Onsager principle \cite{GavassinoSymmetric2022nff,GavassinoUniveraalityI2023odx}, the same reasoning above also explains why their magnetohydrodynamic generalizations were found by \cite{Dash:2022xkz} to undergo plasma oscillations.

\subsection{How many plasmons can there be?}\label{iiic}

Beyond offering a thermodynamic perspective on the origin of the oscillatory behavior, the Onsager analysis above also leads to new rigorous predictions. Specifically, let us revisit the Drude model from section \ref{stepbystep}, and replace the relaxation term $(\langle f \rangle {-} f)/\tau$ in \eqref{boltzmannVlasov} with a more realistic collision integral $C[f]$. How does this modification affect the plasma oscillations? Does the system still exhibit a uniquely defined pair of oscillatory modes, or does a level splitting occur, giving rise to multiple distinct plasmon excitations? Let's see.

If the medium is isotropic, there is no loss of generality in assuming that the electric field is aligned with the $x$ axis at all times (see appendix \ref{appisotropy} for a formal proof). Hence, we can drop the space indices in \eqref{langle}, and simply write
\begin{equation}\label{prima}
\begin{split}
\Dot{\Psi}={}& -\mathcal{L}\Psi +\mathcal{J} E\, , \\
\Dot{E} ={}&-\mathcal{J}^T\Psi \, . \\
\end{split}
\end{equation}
Now, in kinetic theory, the vector $\Psi$ is, strictly speaking, infinite-dimensional. In practice, however, the distribution function $f$ cannot possess more independent moments than there are electrons in the system. Therefore, one may regard $\Psi$ as having a very large but still finite dimension $N$ (of the order of $10^{23}$, in metals). Thus, since $\mathcal{L}$ is real symmetric, there exists an orthogonal transformation $\mathcal{O}$ (with $\mathcal{O}\mathcal{O}^T{=}\mathcal{O}^T \mathcal{O}{=}1$) such that
$\mathcal{L}{=}\mathcal{O}^T \Lambda \mathcal{O}$, with $\Lambda$ diagonal. Then, multiplying the first line of \eqref{prima} by $\mathcal{O}$, and adopting the redefinitions 
$\mathcal{O}\Psi\rightarrow \Psi$ and $\mathcal{O}\mathcal{J}\rightarrow \mathcal{J}$, we obtain
\begin{equation}\label{dopo}
\dfrac{d}{dt}
\begin{bmatrix}
\Psi\\
E\\
\end{bmatrix}
=-
\left[\begin{array}{c|c}
\Lambda & -\mathcal{J} \\
\hline
\mathcal{J}^T & 0 \\
\end{array}\right]
\begin{bmatrix}
\Psi\\
E\\
\end{bmatrix}.
\end{equation}
This allows to invoke the following linear-algebra result.
\begin{theorem}
For any $\Lambda{=}\textup{diag}(\sigma_1,\sigma_2,...,\sigma_N){\in}\, \mathbb{R}^{N\times N}$ and $\mathcal{J}{=}(\mathcal{J}^1,\mathcal{J}^2,...,\mathcal{J}^N)^T{\in}\, \mathbb{R}^N$ (with $N{>}0$), the block matrix
\begin{equation}
Q=
\left[\begin{array}{c|c}
\Lambda & -\mathcal{J} \\
\hline
\mathcal{J}^T & 0 \\
\end{array}\right]
\end{equation}
has at most one couple of complex-conjugate eigenvalues.
\end{theorem}
\begin{proof}
In the following, we will work under the assumption that: (a) $\sigma_n\,{\neq}\, \sigma_m$ for $n\,{\neq}\, m$, and (b) $\mathcal{J}^n \,{\neq}\, 0$ for all $n$. The remaining cases can all be recovered by continuity, since the spectrum of a matrix is a continuous function of the entries \cite[\S 2.5.1]{Kato2008}. 

Our task is to set a lower bound for the number of real roots of the characteristic polynomial
\begin{equation}\label{characteristic}
\det \left[\begin{array}{c|c}
\Lambda-\lambda & -\mathcal{J} \\
\hline
\mathcal{J}^T & -\lambda \\
\end{array}\right] = \left(-\lambda +\sum_n \dfrac{(\mathcal{J}^n)^2}{\sigma_n -\lambda} \right) \prod_m (\sigma_m -\lambda)\, .
\end{equation}
To this end, we focus our attention on the function in the round brackets:
\begin{equation}
g(\lambda)= -\lambda +\sum_n \dfrac{(\mathcal{J}^n)^2}{\sigma_n -\lambda} \, .
\end{equation}
This function has exactly $N$ singularities on the real line, located at $\lambda=\sigma_n$ (since all $\sigma_n$ are distinct, and $\mathcal{J}^n\neq 0$). Near each singularity, $g$ diverges to $+\infty$ as $\lambda\to\sigma_n^-$, and to $-\infty$ as $\lambda\to\sigma_n^+$. Consequently, as we travel from one singularity to the next, $g$ goes from negative to positive, and must cross 0 by continuity ($g$ is smooth for $\lambda\neq \sigma_n$). Since the zeros of $g$ are also roots of \eqref{characteristic}, we conclude that the characteristic polynomial has at least $N-1$ real roots, each lying between two consecutive $\sigma_n$’s. Given that the polynomial is of degree $N+1$, only two roots remain unaccounted for, which may in principle be complex.
\end{proof}

This theorem tells us that there are \textit{at most} two oscillatory solutions of \eqref{dopo} (i.e., two plasmons). All other modes are still forced to sit on the imaginary axis (since $i\omega=$ ``eigenvalue of $Q$'').

In figure \ref{fig:examples}, we provide two examples of ``typical'' non-hydrodynamic spectra of kinetic theories coupled with the electric field. To construct these plots, we consider a model with $N{=}10$ (higher $N$ give similar results), and we set
\begin{equation}\label{randomiamo}
\sigma_n =\dfrac{1}{\tau} \times \text{Random}[0,1]\, , \spc \mathcal{J}^n =\dfrac{g}{\tau} \times \text{Random}[-1,1] \, ,
\end{equation}
where $\tau$ is the typical equilibration timescale, and $\text{Random}[a,b]$ is a random variable uniformly distributed over the interval $[a,b]$. The dimensionless number $g$ quantifies the strength of the coupling between the particles and the electric field. For a given realization of the random variables, we let $g$ run from 0 to 1, which allows us to transition from $\omega_p \tau \ll 1/2$ to $\omega_p\tau \gg 1/2$ (since $\omega_p^2\tau^2 =\mathcal{J}^T\mathcal{J}\tau^2 \sim Ng^2$). The resulting behavior is straightforward: when $g=0$, all poles lie on the imaginary axis, and the electric field is conserved ($\partial_t E=0$), corresponding to a pole at $\omega=0$. As $g$ increases, this pole shifts downward until it collides with one of the slowly relaxing modes. Upon merging, the two poles develop a real part and progressively move away from the imaginary axis. This symmetric couple of emerging modes may be interpreted as the plasmon pair. Note that the plasmons are not necessarily the longest-lived non-hydrodynamic modes. There may be some purely relaxational modes that live longer (see figure \ref{fig:examples}, right panel). 

\begin{figure}[b!]
    \centering
\includegraphics[width=0.43\linewidth]{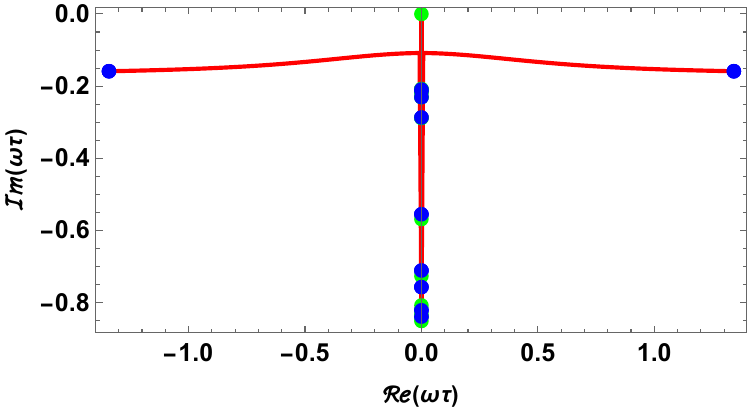}
\includegraphics[width=0.43\linewidth]{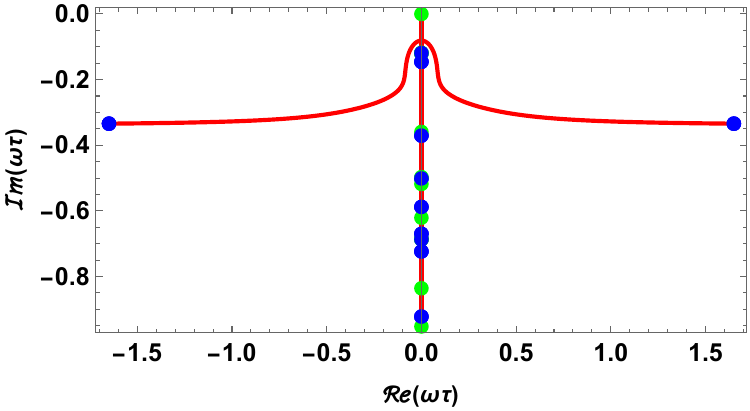}
    \caption{Non-hydrodynamic sector of two Boltzmann-Vlasov-Maxwell models with 10 non-equilibrium degrees of freedom, governed by equation \eqref{dopo}, with $\sigma_n$ and $\mathcal{J}^n$ given by \eqref{randomiamo}. For a given realization of the random variables, we let $g$ run from 0 (green poles) to 1 (blue poles), and we mark the trajectory drawn by the moving poles in red.}
    \label{fig:examples}
\end{figure}





\section{Israel-Stewart theory of electric conduction in metals}
\vspace{-0.3cm}

We have shown that, in the long-wavelength limit, the nonhydrodynamic modes of an electric conductor modeled within Israel-Stewart theory are usually physical, and can be identified with plasma oscillations. At finite wavenumber, however, this correspondence is not guaranteed. In this final section, we examine the behavior of the Israel–Stewart theory for electric conduction in metals (or, more generally, for conductors with immobile positive charges) at finite $k$, highlighting which predictions are robust and which should be treated with caution.

\vspace{-0.4cm}
\subsection{Equations of motion}
\vspace{-0.3cm}

We consider an infinitely large metallic sample at rest, with charge density $\rho(t,x^j)$ and electric current $J^k(t,x^j)$. We assume that bound charges can be neglected and the medium is isotropic. Then, the Israel-Stewart-Maxwell theory of electric conduction is governed by the following equations \cite{Dash:2022xkz}:
\begin{equation}
\begin{split}
& \partial_j E^j=\rho \, ,  \qquad \qquad \qquad \,\, \partial_j B^j=0 \, , \\
& (\nabla \times E)^j =-\partial_t B^j \, , \qquad (\nabla \times B)^j =J^j +\partial_t E^j \, , \\
& \partial_t \rho +\partial_j J^j =0 \, , \quad \quad \quad \quad \! \tau \partial_t J^j +J^j =\sigma(E^j-b^2 \partial^j \rho) \, . \\
\end{split}
\end{equation}
The first two lines are Maxwell's equations. In the third line, we wrote the equation of charge conservation, and the Israel-Stewart relaxation equation for the electric current, accounting for charge diffusion \cite{GavassinoThermoelectric2025bnx}. The diffusion term comes with a prefactor $b^2=d\mu/d\rho$, where $\mu$ is the charge chemical potential. In the following, we will be working in the linear regime, meaning that $\{\tau,\sigma,b\}$ will be treated as background constants.

\vspace{-0.4cm}
\subsection{Magnetic diffusion}
\vspace{-0.3cm}

Electromagnetic perturbations in an isotropic conductor decompose into two independent sectors: transverse modes (characterized by $E^j\perp$ ``wavevector'') and longitudinal modes (characterized by $E^j/ \!\! /$ ``wavevector''). To analyze the transverse sector, we take the curl of the ``curl$\,B$'' Maxwell equation and employ the ``curl$\,E$'' relation, which yields $(\nabla\times J)^k = -(\partial_j \partial^j - \partial^2_t)B^k$. We then take the curl of the Israel-Stewart relaxation equation and recast the result in terms of the magnetic field, thereby obtaining the following partial differential equation:
\begin{equation}\label{transversalIS}
\partial_t B^k =\dfrac{1}{\sigma} (1+\tau \partial_t)(\partial_j \partial^j -\partial^2_t)B^k\, .
\end{equation}
In the slow limit ($\partial_t \rightarrow 0$), the evolution reduces to the standard magnetic-diffusion equation, $\partial_t B^k=\sigma^{-1}\partial_j \partial^j B^k$. In contrast, in the long-wavelength regime ($\partial_j \rightarrow 0$), the dynamics satisfy $(\tau \partial^2_t+\partial_t +\sigma)\partial_t B^k=0$. One branch of the solution space is the trivial mode $B^k=\text{const}$, corresponding to the hydrodynamic excitation. The remaining two branches describe a damped harmonic oscillator with frequency $\omega$ specified by \eqref{omtau}. These nonhydrodynamic excitations represent transverse plasma oscillations: the electric field oscillates in a direction orthogonal to the wavevector, with an amplitude that varies slowly along the wave profile (as implied by $\nabla \times E = -\partial_t B$).

In figure \ref{fig:transwave}, we graph the dispersion relations $\omega(k)$ of \eqref{transversalIS} for both $\omega_p \tau<1/2$ (left panel) $\omega_p\tau>1/2$ (right panel).

\begin{figure}[b!]
    \centering
\includegraphics[width=0.45\linewidth]{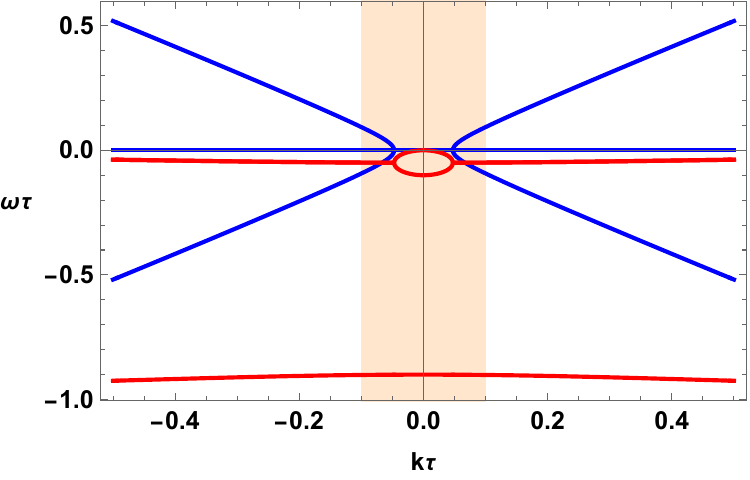}
\includegraphics[width=0.45\linewidth]{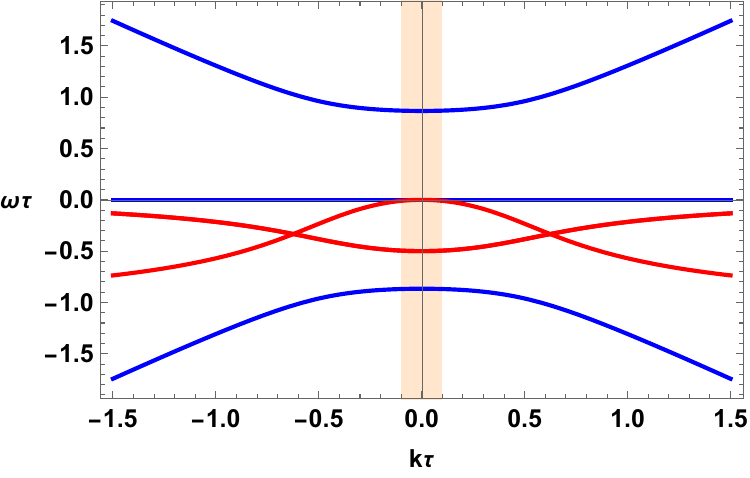}
\caption{Dispersion relations $\omega(k)$ of the transversal modes of a conductor with $\omega_p\tau=0.3$ (left panel) or $1$ (right panel), modeled according to Israel-Stewart theory, see equation \eqref{transversalIS}. The blue lines are the real parts, the red lines are the imaginary parts, and the orange band marks the qualitative regime of validity. Note that there are two relevant scales in \eqref{transversalIS}, namely $\tau$ and $\sigma^{-1}$, but applicability of the theory only requires $k\tau \ll 1$, which is why the regime of validity of the left panel encloses more structure (in the left panel, $\tau \ll \sigma^{-1}$, so the modes with $k\sim \sigma^{-1}$ are physical).}
    \label{fig:transwave}
\end{figure}

\subsection{Debye screening and density equilibration}
\vspace{-0.2cm}

To analyze the longitudinal sector, we take the divergence of the Israel-Stewart relaxation equation. Invoking the conservation of charge and Gauss' law, we obtain the following partial differential equation for the charge density:
\begin{equation}\label{islongitud}
\tau \partial^2_t \rho  +\partial_t \rho +\sigma\rho- \sigma b^2\partial_j  \partial^j \rho=0\, .
\end{equation}
This equation predicts a broad range of physical effects that manifest in diverse settings. For example, in the stationary limit, \eqref{islongitud} reduces to the Debye–Hückel equation, $b^2\partial_j\partial^j\rho=\rho$ \cite[\S 1.6]{bellan_2006}. This relation implies that placing a point charge $Q$ at ${\bf x}{=}0$ inside an infinite conductor produces a screened electric field (with screening length $b$), due to the presence of an induced charge density $\rho(r)=-Q e^{-r/b}/(4\pi b^2 r)$ of opposite sign that surrounds the inserted point charge. The same equation also reproduces the well-known result that an excess charge in a conductor resides near its surface, within a layer of thickness $\sim b$. For example, if a spherical conductor of radius $R$ carries a net charge $Q$, the charge density profile is
\begin{equation}
\rho(r)=\dfrac{Q\sinh(r/b)}{4\pi b r \!\left[R\cosh(R/b)-b\sinh(R/b)\right]} \, .
\end{equation}

If, instead, we restore time-dependence and take the long-wavelength limit ($b\partial_j \to 0$), equation \eqref{islongitud} collapses to the usual damped-oscillator form $\tau \partial_t^2 \rho + \partial_t \rho + \sigma \rho = 0$, whose solutions have frequency $\omega$ given by \eqref{omtau}. These nonhydrodynamic modes describe longitudinal plasma oscillations: the electric field oscillates in a direction parallel to the wavevector, with an amplitude that varies slowly along the propagation direction (as implied by $\partial_j E^j = \rho$).


In figure \ref{fig:longwave}, we graph the dispersion relations $\omega(k)$ of \eqref{islongitud} for both $\omega_p \tau<1/2$ (left panel) $\omega_p\tau>1/2$ (right panel).

\begin{figure}[h!]
    \centering
\includegraphics[width=0.49\linewidth]{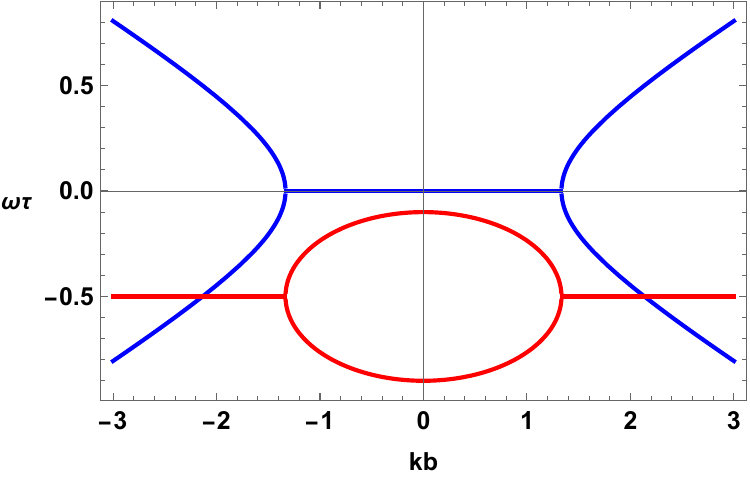}
\includegraphics[width=0.48\linewidth]{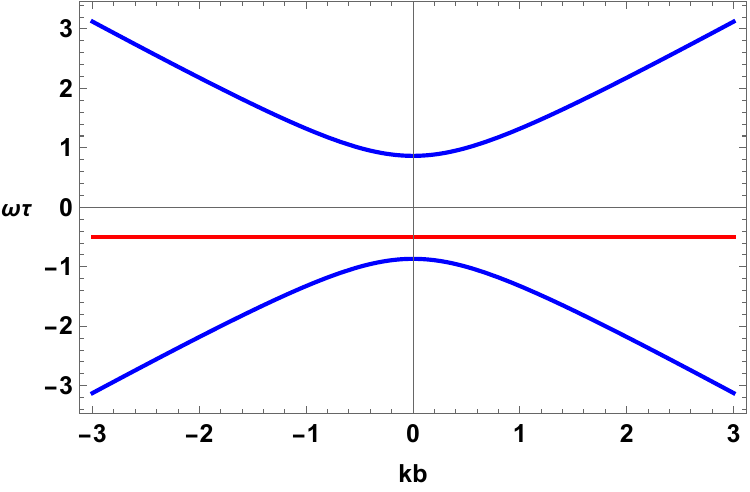}
\caption{Dispersion relations $\omega(k)$ of the longitudinal modes of a conductor with $\omega_p\tau=0.3$ (left panel) or $1$ (right panel), modeled according to Israel-Stewart theory, see equation \eqref{islongitud}. The blue lines are the real parts and the red lines are the imaginary parts. The regime of applicability is defined by the condition $k\tau\ll 1$, which in the plot corresponds to $kb\ll b/\tau$. }
    \label{fig:longwave}
\end{figure}

\vspace{-0.3cm}
\section{Conclusions}
\vspace{-0.3cm}

We have shown that the Israel-Stewart theory of electric conduction, when coupled to Maxwell’s equations, exhibits a substantially broader predictive scope than traditionally assumed. Although conceived as a causal extension of resistive magnetohydrodynamics, the framework reproduces not only the hydrodynamic response of a conducting medium, but also the nonhydrodynamic plasma oscillations encoded in the underlying kinetic theory, in agreement with the expectations of \cite{Denicol_Relaxation_2011,Wagner:2023jgq}. By tracing the nonhydrodynamic sector of Israel-Stewart MHD back to the Boltzmann-Vlasov equations, we established that the oscillatory modes identified in \cite{Dash:2022xkz} are genuine physical plasmons, rather than artifacts of the effective model: they possess the correct frequency, damping rate, and algebraic multiplicity.

The Onsager-Casimir analysis \cite{Onsager_Casimir} provides a unifying thermodynamic interpretation of this result. At zero wavenumber, kinetic theories built solely from PT-even variables admit only purely relaxational modes. However, once the electromagnetic field is incorporated (which is PT-odd), antisymmetric couplings necessarily appear, giving rise to one (and only one) pair of oscillating non-hydrodynamic modes per spatial direction. The Israel-Stewart equations also respect this symmetry structure \cite{GavassinoSymmetric2022nff,GavassinoUniveraalityI2023odx}, and this is why they succeed in reproducing the plasmon physics of the microscopic theory, even though they were not explicitly designed for that purpose.

\newpage
These findings substantially broaden the domain of applicability of Israel–Stewart \textit{magneto}hydrodynamics \cite{Dash:2022xkz,GavassinoShokryMHD:2023qnw}. In particular, the fact that its nonhydrodynamic sector corresponds to physical plasmons implies that the theory remains predictive on timescales of order $\omega_p^{-1}$, far beyond the narrow regime accessible to strictly hydrodynamic frameworks (such as BDNK-type theories \cite{GavassinoThermoelectric2025bnx}). When combined with the ability of Israel–Stewart MHD to capture the correct electrostatic response to gravity and acceleration \cite{Gavassino:2025wcy}, this result firmly establishes the framework as a reliable tool for modeling electromagnetic processes in astrophysical plasmas and heavy-ion collisions.

\section*{Acknowledgements}

This work is supported by a MERAC Foundation prize grant,  an Isaac Newton Trust Grant, and funding from the Cambridge Centre for Theoretical Cosmology.

\appendix
\section{Quick derivation of the Onsager-Casimir relations}\label{apponsager}

First of all, let us note that, if \eqref{identif} holds, then the equal-time correlator matrix of $\Psi$ is just
\begin{equation}\label{a1}
\langle \Psi(0) \Psi^T(0)\rangle =\dfrac{\int e^{-\Psi^T\Psi/2}\, \Psi \Psi^T \, \mathcal{D}\Psi}{\int e^{- \Psi^T \Psi/2} \, \mathcal{D}\Psi}= 1 \, ,
\end{equation}
where ``$1$'' stands for the identity matrix in $\Psi-$space. Next, let us focus on the non-equal time correlator. Since the equilibrium state is invariant under time translations, we have
\begin{equation}\label{psittt}
\langle \Psi(t) \Psi^T(0)\rangle=\langle \Psi(t{-}t) \Psi^T(0{-}t)\rangle=\langle \Psi(0) \Psi^T(-t)\rangle \, .  
\end{equation}
Moreover, invariance of the theory under $\epsilon$ just means that, if a \textit{microscopic} state $X$ evolves into $Y$ after a time $t$, then the transformed state $\epsilon Y$ evolves into $\epsilon X$ in the same time $t$. In one formula, $\epsilon\Phi^{t}=\Phi^{-t} \epsilon$ \cite{Carbone2020}, where $\Phi^t$ is the time-evolution diffeomorphism\footnote{Precisely the same identity holds in a quantum setting. If $\hat{\epsilon}$ is the anti-unitary operator corresponding to $\epsilon$ in the quantum theory, we have that $\Hat{\epsilon}\,e^{-i\Hat{H}t}=\Hat{\epsilon}\,e^{-i\Hat{H}t}\,\Hat{\epsilon}^{-1}\Hat{\epsilon}=e^{-\Hat{\epsilon}i\Hat{H}\Hat{\epsilon}^{-1} t}\Hat{\epsilon}=e^{i\Hat{H} t}\Hat{\epsilon}$, where we used the well-known identity $\hat{\epsilon}i\Hat{H}\hat{\epsilon}^{-1}=-i\hat{H}$, valid for any Hamiltonian that is invariant under the action of $\epsilon$ \cite[\S 2.6]{weinbergQFT_1995}.}. Then, we have the following chain of identities:
\begin{equation}\label{epsiiii}
\begin{split}
\langle \Psi(0) \Psi^T(-t)\rangle ={}& \int \mathcal{P}(X)\Psi(X)\Psi^T(\Phi^{-t} X) dX \\
={}& \int \mathcal{P}(\epsilon Y)\Psi^T(\epsilon Y)\Psi(\Phi^{-t} \epsilon Y) |\text{det}(\epsilon)|dY \\
={}& \int \mathcal{P}(Y)\Psi(\epsilon Y)\Psi^T(\epsilon\Phi^{t} Y) dY \\
={}& \int \mathcal{P}(Y)\epsilon\Psi( Y)\Psi^T(\Phi^{t} Y)\epsilon^T dY \\
={}& \epsilon\langle \Psi(0) \Psi^T(t)\rangle \epsilon^T \, . \\
\end{split}
\end{equation}
In the first line, we expressed the thermal average as the integral over the microscopic phase space with the appropriate ensemble probability $\mathcal{P}(X)$. In the second line, we made the change of variables $X=\epsilon Y$. In the third line, we invoked the invariance of the equilibrium state under $\epsilon$, i.e. $\mathcal{P}(\epsilon Y)=\mathcal{P}(Y)$, the symmetry relation $\Phi^{-t} \epsilon=\epsilon\Phi^{t}$, and the property $\epsilon^2=1$, which implies $|\det \epsilon|= 1$. In the fourth line, we used the fact that the transformation $\Psi \rightarrow \epsilon \Psi$ is \textit{defined} by taking $X\rightarrow \epsilon X$ and looking at its effect on the observables, i.e. $\epsilon \Psi(X)\equiv \Psi(\epsilon X)$ (``pull-back''). Combining \eqref{psittt} and \eqref{epsiiii}, we arrive at
\begin{equation}
\langle \Psi(t) \Psi^T(0)\rangle=  \epsilon\langle \Psi(0) \Psi^T(t)\rangle \epsilon^T \, , 
\end{equation}
which is also valid in a quantum setting \cite{kubo1957statistical}. Finally, Onsager's regression hypothesis allows us to differentiate with respect to $t$ at $t=0^+$, and to use \eqref{dotpsiLpsi} inside the average sign, giving
\begin{equation}
\mathcal{L} \langle \Psi(0) \Psi^T(0)\rangle=  \epsilon\langle \Psi(0) \Psi^T(0)\rangle \mathcal{L}^T \epsilon^T \, .
\end{equation}
Recalling equation \eqref{a1}, we finally obtain equation \eqref{Onsager}.  

\section{The Drude model fulfills Onsager symmetry}\label{drudeonsager}

First of all, we need a formula for $\Delta\Omega$. Applying information-current techniques to a kinetic theory of Fermions \cite{GavassinoGibbs2021,GavassinoCausality2021,SoaresRochaGavassinoKIneticFluctu2024afv}, one finds that
\begin{equation}\label{Omegadrude}
\Delta\Omega = \dfrac{V}{2}\int \dfrac{2d^3 p}{(2\pi)^3} \dfrac{(\delta f)^2}{f_{\text{eq}}(1-f_\text{eq})} +\dfrac{V}{T} \left(\dfrac{E^j E_j}{2}+\dfrac{B^j B_j}{2} \right)\, ,
\end{equation}
where $V$ is the volume, $T$ is the temperature, $f_{\text{eq}}$ is the equilibrium Fermi-Dirac distribution, and $\delta f$ is the linear deviation of $f$ from $f_{\text{eq}}$. In the following, we will set our unit of length so that $V/T=1$. Now, let $\{g_n(p^k)\}_{n=0}^{\infty}$ be a basis of functions from $\mathbb{R}^3$ to $\mathbb{R}$ that are orthonormal relative to the inner product defined by $\Delta \Omega$, i.e.
\begin{equation}
T \int \dfrac{2d^3 p}{(2\pi)^3} \dfrac{g_m g_n}{f_{\text{eq}}(1-f_\text{eq})} =\delta_{mn} \, ,
\end{equation}
and expand $\delta f$ on this basis, $\delta f=\sum_n \Psi^n g_n$, where $\Psi^n$ are some linear-combination coefficients. Then, \eqref{Omegadrude} becomes
\begin{equation}
\Delta\Omega =\sum_{n} \dfrac{(\Psi^n)^2}{2}  + \dfrac{E^j E_j}{2}+\dfrac{B^j B_j}{2} \, .
\end{equation}
This shows that the infinite vector $\Psi=(\Psi^1,\Psi^2,...)^T$ serves as an appropriate choice of primary degrees of freedom for the application of the Onsager principle. Let us now write the equations of motion of the Boltzmann-Vlasov-Maxwell system linearised around a (non-magnetised) equilibrium state, also accounting for the Lorentz force:
\begin{equation}
\begin{split}
\delta \Dot{f}={}& \dfrac{\langle \delta f\rangle -\delta f}{\tau} +e[E^j+(v{\times} B)^j] \dfrac{\partial f_{\text{eq}}}{\partial p^j}\, , \\ 
\Dot{E}^j ={}& e \int \dfrac{2d^3 p}{(2\pi)^3} \,  \delta f \, v^j \, ,\\
\Dot{B}^j ={}& 0 \, .\\
\end{split}
\end{equation}
Expanding $\delta f$ in the basis $g_n$, and invoking the Fermi-Dirac identity $\partial f_{\text{eq}}/\partial p^j=-f_{\text{eq}}(1-f_\text{eq}) v_j/T$, we obtain
\begin{equation}
\begin{split}
\sum_n \Dot{\Psi}^n g_n={}& \sum_n \dfrac{\langle g_n\rangle -g_n}{\tau} \Psi^n -eE^j\dfrac{ v_j}{T} f_{\text{eq}}(1-f_\text{eq}) \, , \\ 
\Dot{E}^j ={}& \sum_n \Psi^n e \int \dfrac{2d^3 p}{(2\pi)^3} \,  g_n \, v^j \, ,\\
\Dot{B}^j ={}& 0 \, .\\
\end{split}
\end{equation}
Note that the magnetic field has disappeared from the Vlasov term, as expected. Finally, let us take the inner product of the first line with $g_m$, and we arrive at
\begin{equation}
\begin{split}
\Dot{\Psi}^m ={}& \sum_n T \int \dfrac{2d^3 p}{(2\pi)^3} \dfrac{\langle g_m\rangle \langle g_n\rangle -g_m g_n }{f_{\text{eq}}(1-f_\text{eq})\tau}  \Psi^n -eE^j \int \dfrac{2d^3 p}{(2\pi)^3} g_m v_j  \, , \\ 
\Dot{E}^j ={}& \sum_n \Psi^n e \int \dfrac{2d^3 p}{(2\pi)^3} \,  g_n \, v^j \, ,\\
\Dot{B}^j ={}& 0 \, ,\\
\end{split}
\end{equation}
which is in perfect agreement with \eqref{langle}, under the following identifications:
\begin{equation}
\begin{split}
\mathcal{L}_{mn}={}& T \int \dfrac{2d^3 p}{(2\pi)^3} \dfrac{(g_m{-}\langle g_m\rangle) (g_n{-} \langle g_n\rangle)}{f_{\text{eq}}(1-f_\text{eq})\tau} \spc (\text{as expected, }\mathcal{L}=\mathcal{L}^T\text{ and }\mathcal{L}\succeq 0)\, , \\
\mathcal{J}^n_j ={}& -e  \int \dfrac{2d^3 p}{(2\pi)^3} g_n v_j  \, . \\
\end{split}
\end{equation}
As a last consistency check, let us verify that the identity $\omega_p^2=\sum_n(\mathcal{J}_1^n)^2$ holds. To this end, let us first decompose $\partial f_{\text{eq}}/\partial p^1$ on the $g_n$ basis:
\begin{equation}
\dfrac{\partial f_{\text{eq}}}{\partial p^1} = \sum_n g_n  T\int \dfrac{2d^3 p'}{(2\pi)^3} \dfrac{g_n'}{f'_{\text{eq}}(1-f_\text{eq}')} \dfrac{\partial f'_\text{eq}}{\partial p'^1} = - \sum_n g_n  \int \dfrac{2d^3 p'}{(2\pi)^3} g_n'  v'_1=\dfrac{1}{e}\sum_n \mathcal{J}_1^n g_n\, .
\end{equation}
Then, we have that
\begin{equation}
\omega_p^2 = \dfrac{e^2}{m} n_e =\dfrac{e^2}{m}\int \dfrac{2d^3 p}{(2\pi)^3} f_{\text{eq}}=e^2\int \dfrac{2d^3 p}{(2\pi)^3} f_{\text{eq}}\dfrac{\partial v^1}{\partial p^1}=-e^2\int \dfrac{2d^3 p}{(2\pi)^3} \dfrac{\partial f_{\text{eq}}}{\partial p^1} v^1=-\sum_n \mathcal{J}_1^n e\int \dfrac{2d^3 p}{(2\pi)^3} g_n v^1=\sum_n (\mathcal{J}_1^n)^2 \, .
\end{equation}

\section{Consequences of isotropy}\label{appisotropy}

\begin{theorem}
If the background equilibrium state is isotropic, any solution $\Psi_{\text{BVM}}(t)$ of \eqref{langle} can be decomposed as the sum of three solutions $\Psi_{\text{BVM}}^{(j)}(t)$ such that $E_k^{(j)}(t)\propto \delta^j_k$.
\end{theorem}
\begin{proof}
Since the magnetic field decouples from the other variables, we will set it to zero for clarity.

Let
\begin{equation}
\Psi_{\text{BVM}}(t)= \left\{\Psi(t),E^x(t),E^y(t),E^z(t)\right\}^T 
\end{equation}
be an arbitrary solution of \eqref{langle}. If the background equilibrium state is isotropic, then any rotated version of $\Psi_{\text{BVM}}(t)$ is also a solution of \eqref{langle}. Hence, let $\Psi_{\text{BVM}}^{(x)}(t)$, $\Psi_{\text{BVM}}^{(y)}(t)$, and $\Psi_{\text{BVM}}^{(z)}(t)$ be the solutions arising from applying 180$^{\text{o}}$-rotations around the $x$, $y$, and $z$, respectively, to $\Psi_{\text{BVM}}(t)$. Clearly, we have that
\begin{equation}
\begin{split}
\Psi_{\text{BVM}}^{(x)}(t)={}& \left\{\Psi^{(x)}(t),+E^x(t),-E^y(t),-E^z(t)\right\}^T \, , \\
\Psi_{\text{BVM}}^{(y)}(t)={}& \left\{\Psi^{(y)}(t),-E^x(t),+E^y(t),-E^z(t)\right\}^T \, , \\
\Psi_{\text{BVM}}^{(z)}(t)={}& \left\{\Psi^{(z)}(t),-E^x(t),-E^y(t),+E^z(t)\right\}^T \, . \\
\end{split}
\end{equation}
But since the system \eqref{langle} is linear, any linear combination of $\Psi_{\text{BVM}}(t)$, $\Psi_{\text{BVM}}^{(x)}(t)$, $\Psi_{\text{BVM}}^{(y)}(t)$, and $\Psi_{\text{BVM}}^{(z)}(t)$ is also a solution. Thus, we can define the following three new solutions of \eqref{langle}:
\begin{equation}
\begin{split}
\Psi_{\text{BVM}}^{(1)}(t)={}& \dfrac{\Psi_{\text{BVM}}(t)+ \Psi_{\text{BVM}}^{(x)}(t)}{2} \, , \\
\Psi_{\text{BVM}}^{(2)}(t)={}& \dfrac{\Psi_{\text{BVM}}(t)- \Psi_{\text{BVM}}^{(x)}(t)+\Psi_{\text{BVM}}^{(y)}(t)- \Psi_{\text{BVM}}^{(z)}(t)}{4} \, , \\
\Psi_{\text{BVM}}^{(3)}(t)={}& \dfrac{\Psi_{\text{BVM}}(t)- \Psi_{\text{BVM}}^{(x)}(t)-\Psi_{\text{BVM}}^{(y)}(t)+ \Psi_{\text{BVM}}^{(z)}(t)}{4} \, . \\
\end{split}
\end{equation}
It is immediate to see that $\Psi_{\text{BVM}}(t)=\Psi_{\text{BVM}}^{(1)}(t)+\Psi_{\text{BVM}}^{(2)}(t)+\Psi_{\text{BVM}}^{(3)}(t)$. Furthermore,
\begin{equation}
\begin{split}
\Psi_{\text{BVM}}^{(1)}(t)={}& \left\{\Psi^{(1)}(t),E^x(t),0,0\right\}^T \, , \\
\Psi_{\text{BVM}}^{(2)}(t)={}& \left\{\Psi^{(2)}(t),0,E^y(t),0\right\}^T \, , \\
\Psi_{\text{BVM}}^{(3)}(t)={}& \left\{\Psi^{(3)}(t),0,0,E^z(t)\right\}^T \, . \\
\end{split}
\end{equation}
This completes our proof.    
\end{proof}
\textbf{Remark:} Note that, in applications to solids, the assumption of isotropy can be significantly relaxed. Indeed, our proof only relies on invariance under 180$^{\text{o}}$-rotations about the coordinate axes. Therefore, the result also holds for materials with an underlying cubic lattice structure, as long as the coordinate axes are aligned with the lattice edges.

\bibliography{Biblio}

\label{lastpage}
\end{document}